\documentclass[conference, 10pt]{IEEEtran}
\IEEEoverridecommandlockouts %this enabbles thanks in the conference mode

\usepackage[cmex10]{amsmath}  %cmex10 ensures that only type 1 fonts will be used
\usepackage{amssymb}
\usepackage{graphicx,color}   %color for pdf_t generated by xfig

\usepackage{ifxetex}
\ifxetex
\usepackage{xunicode}% unicode character macros
\else
\usepackage[utf8]{inputenc}
\usepackage[T1]{fontenc}
\usepackage{microtype}
\fi

\newtheorem{theorem}{Theorem}
\newtheorem{remark}{Remark}
\newtheorem{lemma}{Lemma}

\begin{document}

%---------- Title ----------
\title{Capacity Results for Two Classes of\\ Three-Way Channels}
\author{\IEEEauthorblockN{Lawrence Ong}
\IEEEauthorblockA{School of Electrical Engineering and Computer Science, 
The University of Newcastle, Australia\\
Email: lawrence.ong@cantab.net}
\thanks{Lawrence Ong is the recipient of an Australian Research Council Discovery Early Career Researcher Award  (project number DE120100246).}
}
\maketitle

\begin{abstract}
This paper considers the three-way channel, consisting of three nodes, where each node broadcasts a message to the two other nodes. The capacity of the {\em finite-field} three-way channel is derived, and is shown to be achievable using a non-cooperative scheme without feedback. The same scheme is also shown to achieve the {\em equal-rate} capacity (when all nodes transmit at the same rate) of the {\em sender-symmetrical} (each node receives the same SNR from the other two nodes) phase-fading AWGN channel. In the light that the non-cooperative scheme is not optimal in general, a cooperative feedback scheme that utilizes relaying and network coding is proposed and is shown to achieve the equal-rate capacity of the {\em reciprocal} (each pair of nodes has the same forward and backward SNR) phase-fading AWGN three-way channel.
\end{abstract}

\section{Introduction}
The two-way channel~\cite{shannon61} models the scenario where two nodes exchange messages via a noisy channel. Despite its simple setup, the capacity (in the Shannon sense) of this channel is not known in general. One of the difficulties is to determine the optimal way to utilize feedback, i.e., how (and if) the node should process and transmit its past received signals. It has been shown that not utilizing the feedback can be strictly suboptimal~\cite{dueck79}. However, for the white additive Gaussian noise (AWGN) two-way channel, Han~\cite{han84} has shown that the capacity can be achieved without using feedback, i.e., using independent Gaussian codewords for the nodes.

The result that feedback is not useful in the AWGN two-way channel does not generalize to the case of more nodes. When there are more nodes, feedback can improve the achievable rates over non-feedback schemes via (i) coherent combining and (ii) relaying~\cite{covergamal79}. Coherent combining is achieved by correlating the nodes' inputs to achieve a higher signal-to-noise ratio (SNR); relaying is achieved by having a node to help in forwarding data when the quality of the direct link from the source to the destination is poor.

In a three-node network, different variations of message flows can take place~\cite{meulen71,meulen77}. These include the multiple-access channel~\cite{ahlswede71}, the relay channel~\cite{covergamal79}, and the broadcast channel~\cite{cover75b}---the capacity of the last two channels remains unknown to date. In this paper, we consider the {\em conferencing} three-way channel (subsequently referred to as the three-way channel), where each node broadcasts its message to the other two nodes.

For the AWGN three-way channel, Eswaran and Gastpar~\cite{eswarangastpar08} proposed a feedback scheme based on the modulated estimate correction (MEC) technique~\cite{kramer02} to exploit the coherent-combining gain. They showed that their proposed scheme achieves the sum-rate capacity under the following assumptions: (i) all nodes receives the same channel output; (ii) the transmit signals are subject to per-symbol power constraints; (iii) all nodes have the same transmitted power, receiver noise power, and channel gain.

In this paper, we consider two classes of three-way channels: the finite-field model and the phase-fading AWGN model. We show that the non-cooperative coding scheme (which does not use feedback) is optimal for the finite-field three-way channel. While the same scheme is not necessarily optimal for the phase-fading AWGN model in general, we show that it is optimal when (i) all nodes transmit at the same rate ({\em equal rate}) and (ii) each node receives an equal SNR from the other two nodes (which we term {\em sender symmetrical}).

We next propose a scheme that utilizes feedback. While the MEC scheme exploits coherent combining, random phase shifts in the phase-fading channel prevent any coding scheme to harvest the gain from coherent combining. Instead, we propose a scheme that utilizes feedback via relaying and network coding. The idea is for one node to help the other two nodes by relaying their messages for each other, and it does so using network coding. We show that this proposed scheme achieves the equal-rate capacity of the {\em reciprocal} (meaning that each node pair has the same SNR in the forward and backward links) phase-fading AWGN three-way channel.

The three-way channel with correlated sources has been studied by Lai et al.~\cite{lailiugamal06}.

%Other variations of the three-way channel considered include three-way channel with correlated sources~\cite{lailiugamal06}, the general three-way channel with arbitrary message flow~\cite{meulen71}. The capacity of these channels remain an open problem.

\subsection{Main Results}

The main results of this paper are
\begin{enumerate}
\item the capacity of the finite-field three-way channel,
\item the equal-rate capacity of the sender-symmetrical phase-fading AWGN three-way channel, and
\item the equal-rate capacity of the reciprocal phase-fading AWGN three-way channel.
\end{enumerate}

\section{Channel Model}

\begin{figure}[t]
\centering
\includegraphics[width=7.5cm]{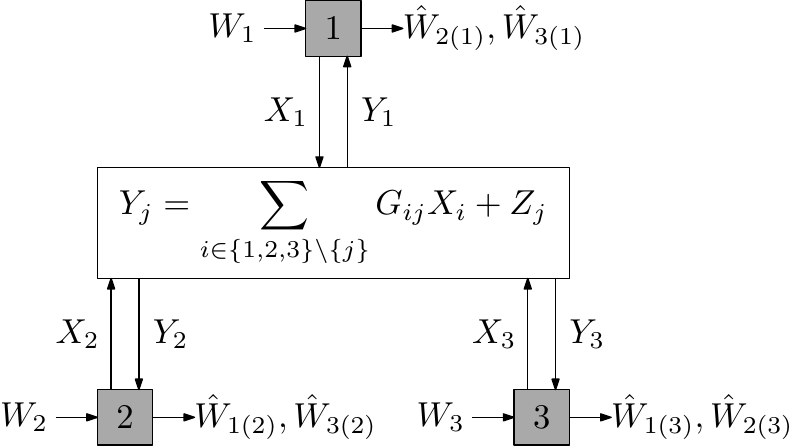}
\caption{The conferencing three-way channel}
\label{fig:three-way-channel}
\end{figure}

Fig.~\ref{fig:three-way-channel} depicts the three-way channel considered in this paper. The nodes in the three-way channel are denoted by nodes 1, 2, and 3. Further denote the message of node $i$, the channel input from node $i$, and the channel output received by node $i$ be $W_i$, $X_i$, and $Y_i$ respectively, for each $i \in \{1,2,3\}$. The three-way channel is defined as follows:
\begin{equation}
Y_j = \sum_{i \in \{1,2,3\} \setminus \{j\}} G_{ij} X_i + Z_j,
\end{equation}
for all $j \in \{1,2,3\}$, where $G_{ij}$ is the channel gain, and $Z_j$ is the receiver noise at node $j$, which is independent for each receiver $j$ and is independent and identically distributed for each channel use.

We consider two models of the three-way channel.
For the finite-field channel, we have the following:
\begin{enumerate}
\item The variables are elements from a finite field, i.e., $X_i,Y_i,Z_i,G_{ij} \in \{0,1,\dotsc, F-1\}$, where $F$ is a prime power. % = a^b$ for some prime number $a$ and some positive integer $b$.
\item The addition and the multiplication are the modulo-$F$ addition and multiplication respectively.
\item $Z_j$ is an arbitrarily distributed random variable.
\item The channel gain $G_{ij}$ is non-zero, fixed, known to all nodes a priori.
\end{enumerate}

For the phase-fading AWGN channel, we have the following:
\begin{enumerate}
\item The variables $X_i,Y_i,Z_i, G_{ij}$ are complex numbers.
\item The addition and the multiplication are complex addition and multiplication respectively.
\item $Z_j$ is circularly symmetric complex Gaussian noise with variance $E[|Z_j|^2] = N_j$.
\item $G_{ij} = |G_{ij}| e^{\imath\theta_{ij}}$ where
\begin{enumerate}
\item the magnitude $|G_{ij}|$ is non-zero, fixed\footnote{Fixed attenuation models static networks where the path loss components are fixed.}, and known to all nodes a priori,
\item the phase shift $\theta_{ij}$ is a random variable\footnote{Random phase shifts prevent coherent combining at receivers. Coherent combining is practically difficult to achieve as it requires synchronizing the high-frequency carrier of two or more signals arriving at each receiver~\cite[p.~2021]{hostmadsenzhang05}.} that is independent and uniformly distributed over $[0,2\pi)$ for each $(i,j)$ pair and for each channel use.  In addition, $\theta_{ij}$ is known to only the receiver, i.e., node $j$ for all $i$. % Note that the phase shift is not reciprocal, i.e., $\theta_{ij}$ might not equal $\theta_{ji}$.
\end{enumerate}
\end{enumerate}

The channel is used multiple times, and we use the square brackets to indicate the random variables associated with each channel use, e.g., $X_1[t]$ denotes $X_1$ on the $t$-th channel use. 
We consider the $(n,R_1,R_2,R_3)$ block code consisting of
\begin{itemize}
\item an independent message for each user $i$: $W_i \in \{1,2,\dotsc, 2^{nR_i}\}$;
\item a set of encoding functions for each user $i$: $X_i[t]= f_{it}(W_i,Y'_t[1],Y'_i[2],\dotsc, Y'_i[t-1])$, where
\begin{equation*}
Y'_i[t] \triangleq \begin{cases}
Y_i[t], & \text{for the finite-field channel} \\
(Y_i[t], \theta_{ji}[t], \theta_{ki}[t]), & \text{for the phase-fading}\\
& \text{AWGN channel}\footnotemark[3],
\end{cases}
\end{equation*}
for $i \neq j \neq k$;\footnotetext[3] {This models the fact that each receiver knows the phase shifts.}\stepcounter{footnote}
\item  a decoding function for each user $i$: $(\hat{W}_{j(i)},\hat{W}_{k(i)}) = d_i(W_i,Y'_i[1],Y'_i[2], \dotsc,Y'_i[n])$, for $i \neq j \neq k$,  where $\hat{W}_{j(i)}$ and $\hat{W}_{k(i)}$ are the estimates of $W_j$ and $W_k$ by node $i$. 
\end{itemize}
Here, $n$ is the codelength, and $R_i$ is the rate at which node $i$ transmits.

For the AWGN channel, we impose an additional average transmitted power constraint on each node: $\sum_{t=1}^n E[|X_i[t]|^2]/n \leq P_i$, where the expectation is taken over the messages, the channel noise, and phase shifts. Such constraint is not imposed on the channel inputs for the case of finite field. We can then define the signal-to-noise ratio (SNR) for the link $i \rightarrow j$ as $\gamma_{ij} \triangleq |G_{ij}|P_i/N_j$. %In this paper, we consider reciprocal AWGN channel where $\gamma_{ij} = \gamma_{ji}$.

Assuming that each $W_i$ is uniformly distributed on $\{1,2,\dotsc, 2^{nR_i}\}$, we define the error probability as $P_\text{e} = \Pr \{ \hat{W}_{j(i)} \neq W_j \text{ for some } i,j \in \{1,2,3\} \text{ where }i \neq j\}$.
The rate triplet $(R_1,R_2,R_3)$ is said to be achievable if the following is true: For any $\epsilon > 0$, there exists for sufficiently large $n$ an $(n,R_1,R_2,R_3)$ code such that $P_\text{e} \leq \epsilon$. The capacity region is the closure of all achievable triplets.

\section{Outer Bounds to the Capacity Region}

We first derive capacity outer bounds.
Using the cut-set argument~\cite[Thm.~15.10.1]{coverthomas06}, if the rate triplet $(R_1,R_2,R_3)$ is achievable, there must exists some joint input distribution $p(x_1,x_2,x_3)$ such that
\begin{align}
R_i &\leq I(X_i; Y'_j,Y'_k|X_j,X_k) \label{eq:upper-2} \\
R_i + R_j &\leq I(X_i,X_j;Y'_k), \label{eq:upper}
\end{align}
for all $i \neq j \neq k$.

\subsection{For the Finite-Field Model}

Let $\mathcal{C}_\text{ff}$ denote the capacity region of the finite-field three-way channel.
The right-hand side of \eqref{eq:upper} is maximized simultaneously for all $i \neq j \neq k$ using independent and uniform distribution~\cite[Sec.~III]{ongmjohnsonit11} for $X_1$, $X_2$, and $X_3$. Evaluating the mutual information term in \eqref{eq:upper}, we have the following capacity outer bound:
\begin{lemma} \label{lemma:finite-field-outer}
$\mathcal{C}_\text{ff} \subseteq \mathcal{R}_1^\text{outer}$, where $\mathcal{R}_1^\text{outer}$ comprises all non-negative triplets $(R_1,R_2,R_3)$ satisfying
\begin{equation}
R_i + R_j \leq \log_2 F - H(Z_k),
\end{equation}
for all $i \neq j \neq k$.
\end{lemma}

We will show that only constraint \eqref{eq:upper} is required for the characterization of $\mathcal{C}_\text{ff}$. %finite-field case.

\subsection{For the Phase-Fading AWGN Model}

Let $\mathcal{C}_\text{awgn}$ denote the capacity region of the phase-fading AWGN three-way channel.
As phase fading (known only to the receivers) inhibits coherent combining at the receivers, the right-hand side of both \eqref{eq:upper-2} and \eqref{eq:upper} is maximized simultaneously for all $i \neq j \neq k$ using independent, zero-mean Gaussian~\cite[Sec.~VII.D]{kramergastpar04}, \cite[Lemma~1]{hostmadsenzhang05} distribution $X_\ell$ with variance $E[|X_\ell|^2] = P_\ell$ for each $\ell \in \{ 1,2,3\}$. Evaluating the mutual information terms in \eqref{eq:upper-2} and \eqref{eq:upper}, we have the following capacity outer bound:
\begin{lemma} \label{lemma:awgn-outer}
$\mathcal{C}_\text{awgn} \subseteq \mathcal{R}_2^\text{outer}$, where $\mathcal{R}_2^\text{outer}$ comprises all non-negative triplets $(R_1,R_2,R_3)$ satisfying
\begin{align}
R_i &\leq \log_2 ( 1 + \gamma_{ij} + \gamma_{ik}) \label{eq:awgn-outer-2} \\
R_i + R_j & \leq \log_2 ( 1 + \gamma_{ik} + \gamma_{jk}), \label{eq:awgn-outer}
\end{align}
for all $i \neq j \neq k$.
\end{lemma}

\section{Capacity Region of the Finite-Field Model}

In this section, we prove the following capacity result:
\begin{theorem} \label{theorem:finite-field-capacity}
The capacity region of the finite-field three-way channel is $\mathcal{C}_\text{ff} = \mathcal{R}_1^\text{outer}$, where $\mathcal{R}_1^\text{outer}$ is defined in Lemma~\ref{lemma:finite-field-outer}.
\end{theorem}

\begin{proof}[Proof of Theorem~\ref{theorem:finite-field-capacity}]
Lemma~\ref{lemma:finite-field-outer} gives an outer bound to $\mathcal{C}_\text{ff}$. So, we only need to show that all triplets in the interior of $\mathcal{R}_1^\text{outer}$ is achievable. We generate the codewords for each node as follows: For user $i$, we randomly generate $2^{nR_i}$ sequences $\boldsymbol{x}_i= (x_i[1], x_i[2], \dotsc, x_i[n])$ according to $p(\boldsymbol{x}_i) = \prod_{t=1}^np^\text{u}(x_i[t])$, where $p^\text{u}(\cdot)$ denotes the uniform distribution. We index each codeword by $\boldsymbol{x}_i(w_i)$, $w_i \in \{1,2,\dotsc, 2^{nR_i}\}$. Each user $i$ transmits $\boldsymbol{X}_i(W_i)$. Each user $k$ receives $\boldsymbol{Y}_k = (Y_k[1], Y_k[2], \dotsc, Y_k[n])$, where $Y_k[t] = G_{ik} X_i[t]+ G_{jk} X_j[t] + Z_k[t]$ is a multiple-access channel. It follows that node $k$ can reliably (with arbitrarily small error probability) decode $W_i$ and $W_j$ if~\cite[Thm.~15.3.4]{coverthomas06}
\begin{align}
R_i &< I(X_i;Y_k|X_j) = H(Y_k|X_j) - H(Y_k|X_i,X_j) \nonumber \\ &= H(G_{ik}X_i + Z_k) - H(Z_k) \nonumber \\ & =  \log_2 F - H(Z_k)  \label{eq:finite-field-lower-1} \\
R_j &< I(X_j;Y_k|X_i) = \log_2 F - H(Z_k)  \label{eq:finite-field-lower-2} \\
R_i + R_j &< I(X_i,X_j;Y_k) = \log_2 F - H(Z_k). \label{eq:finite-field-lower-3}
\end{align}
Note that there is a bijective mapping between $G_{ik} X_i[t]$ and $X_i[t]$ since $G_{ik}$ is not zero, and hence each $G_{ik}X_i[t]$ is also independently and uniformly distributed. Since \eqref{eq:finite-field-lower-3} implies \eqref{eq:finite-field-lower-1} and \eqref{eq:finite-field-lower-2}, repeating the same decoding argument for the other two nodes, we have that all triplets in the interior of $\mathcal{R}_1^\text{outer}$ is achievable. This proves Theorem~\ref{theorem:finite-field-capacity}.
\end{proof}

The capacity region for the finite-field three-way channel can be attained without cooperation, i.e., letting each node transmit as in an point-to-point channel, and without utilizing feedback, i.e., letting each node transmits a function of only its message (not it's received signals).

\section{Capacity Inner Bound for the Phase-Fading AWGN Model} \label{section:awgn}

Using the same coding strategy that achieves the capacity region of the finite-field channel, we have the following:
\begin{lemma} \label{lemma:awgn-lower}
The region $\mathcal{R}_1^\text{inner}$ is achievable on the phase-fading AWGN three-way channel, where $\mathcal{R}_1^\text{inner}$ comprises all non-negative triplets satisfying
\begin{align}
R_i &< \log_2 ( 1 + \gamma_{ik}) \label{eq:awgn-lower-1}\\
R_j &< \log_2 ( 1 + \gamma_{jk}) \label{eq:awgn-lower-2}\\
R_i + R_j &< \log ( 1 + \gamma_{ik} + \gamma_{jk}), \label{eq:awgn-lower-3}
\end{align}
for all $i \neq j \neq k$.
\end{lemma}

\begin{proof}[Proof of Lemma~\ref{lemma:awgn-lower}]
% We use the same coding strategy as that for the finite-field channel.
We randomly generate the codewords $\boldsymbol{x}_i(w_i)$ for each $i \in \{1,2,3\}$ and for each $w_i  \in \{1,2,\dotsc,2^{nR_i}\}$, where each codeletter $x_i[t](w_i)$ is independently generated according to the circularly symmetric complex Gaussian distribution with variance $E[|X_i|^2] = P_i$. Lemma~\ref{lemma:awgn-lower} follows from the achievability results of the multiple-access channel~\cite[Thm.~15.3.4]{coverthomas06} and the phase-fading complex Gaussian channel~\cite[Lemma~1]{hostmadsenzhang05}.
\end{proof}

It follows from the above lemma that the capacity region of the phase-fading AWGN three-way channel is inner bounded as $\mathsf{cl}(\mathcal{R}_1^\text{inner}) \subseteq \mathcal{C}_\text{awgn} $, where $\mathsf{cl}(\cdot)$ denotes the closure operator.
We note that \eqref{eq:awgn-lower-1} and \eqref{eq:awgn-lower-2} are both stricter than \eqref{eq:awgn-outer-2}. Hence, $\mathsf{cl}(\mathcal{R}_1^\text{inner})$ might not equal $\mathcal{R}_2^\text{outer}$. Unlike the finite-field channel where the non-cooperative scheme achieves the capacity region, this scheme may not be optimal for the phase-fading AWGN channel.

\begin{remark}
We have used independent circularly symmetric complex Gaussian codewords~\cite{telatar99} here to maximize the mutual information terms $I(X_i;Y_k|X_j)$ and $I(X_i,X_j;Y_k)$ to obtain \eqref{eq:awgn-lower-1}--\eqref{eq:awgn-lower-3}. However, as the transmitted signals go through random phase shifts, the same results \eqref{eq:awgn-lower-1}--\eqref{eq:awgn-lower-3} can also be obtained using independent real Gaussian codewords with the same variance.
\end{remark}

\section{Equal-Rate Capacity of the Phase-Fading AWGN Model}

In this section, we simplify the problem by considering the equal-rate points, i.e., where $R_1 = R_2 = R_3$. An equal rate $R$ is said to be achievable if and only if $(R,R,R)$ is achievable. This scenario guarantees fairness of the nodes by ensuring that all nodes are able to transmit at the same rate. In this setting, we define the {\em equal-rate capacity} as the supremum of the set of achievable equal rates. The equal-rate capacity is the maximum rate at which all nodes can transmit simultaneously.

%Although the capacity region of the phase-fading AWGN three-way channel remains unsolved, we will derive the equal-rate capacity for two classes of phase-fading AWGN three-way channels.
We now derive the equal-rate capacity for two classes of phase-fading AWGN three-way channels.

% \subsection{The Symmetrical Case}

% Consider a special class, namely, the {\em symmetrical} phase-fading AWGN channel where the SNRs of all links are equal, i.e., $\gamma_{ij} = \gamma$, for all $i, j \in \{1,2,3\}$ where $i \neq j$.

% \begin{theorem} \label{theorem:symmetrical}
% \end{theorem}

% \begin{proof}[Proof of Theorem~\ref{theorem:symmetrical}]
% For the symmetrical case, both \eqref{eq:awgn-outer} for the outer bound and \eqref{eq:awgn-lower-3} for the inner bound simplify to $R_i + R_j \leq \log_2 ( 1 + 2 \gamma)$, for all $i \neq j$. In addition, 
% \end{proof}

\subsection{The Sender-Symmetrical Case}

Consider the first class of phase-fading AWGN channels where the SNR on the link $i \rightarrow k$ equals that on $j \rightarrow k$, i.e., $\gamma_{ik} = \gamma_{jk} \triangleq \gamma_k$ for each $k$. We term this the {\em sender-symmetrical case} as each receiver sees an equal effective transmitted power (after channel attenuation) from the other two transmitters. Note that $\gamma_j$ might not equal $\gamma_k$. For this class of three-way channel, we have the following:
\begin{theorem} \label{theorem:sender-symmetrical}
The equal-rate capacity of the sender-symmetrical phase-fading AWGN three-way channel is 
\begin{equation}
C_\text{ss} = \frac{1}{2} \log_2 \left( 1 + 2 \min_{k \in \{1,2,3\}}\gamma_k \right). \label{eq:sender-symmetrical-capacity}
\end{equation}
\end{theorem}

\begin{proof}[Proof of Theorem~\ref{theorem:sender-symmetrical}]
From Lemma~\ref{lemma:awgn-outer}, we know that if the equal rate $R$ is achievable, then $R \leq \log_2 ( 1 + \gamma_j + \gamma_k)$ and $R \leq \frac{1}{2}\log_2 ( 1 + 2 \gamma_k)$ for all $j, k \in \{1,2,3\}$ where $j \neq k$. Hence, the equal-rate capacity is upper bounded as $C_\text{ss} \leq \frac{1}{2} \log_2 \left( 1 + 2 \min_{k \in \{1,2,3\}}\gamma_k \right)$.

\begin{table*}
\renewcommand{\arraystretch}{1.5}
\caption{A Cooperative Coding Scheme using Feedback}
\label{table}
\centering
\begin{tabular}{c||c c c c c c}
\bfseries Block & 1 & 2 & 3 & $\dotsm$ & $B$ & $B+1$ \\
\hline
Node 2 transmits & $\boldsymbol{X}_2(W_2^{(1)})$ & $\boldsymbol{X}_2(W_2^{(2)})$ & $\boldsymbol{X}_2(W_2^{(3)})$ & $\dotsm$ & $\boldsymbol{X}_2(W_2^{(B)})$ & $\boldsymbol{X}_2(0)$\\
Node 3 transmits & $\boldsymbol{X}_3(W_3^{(1)})$ & $\boldsymbol{X}_3(W_3^{(2)})$ & $\boldsymbol{X}_3(W_3^{(3)})$ & $\dotsm$ & $\boldsymbol{X}_3(W_3^{(B)})$ & $\boldsymbol{X}_2(0)$\\
Node 1 decodes & $\hat{W}_2^{(1)},\hat{W}_3^{(1)} \;\rightarrow$ & $\hat{W}_2^{(2)},\hat{W}_3^{(2)} \;\rightarrow$ & $\hat{W}_2^{(3)},\hat{W}_3^{(3)} \;\rightarrow$ & $\dotsm $ & $\hat{W}_2^{(B)},\hat{W}_3^{(B)}$ & --- \\
Node 1 transmits & $\boldsymbol{X}_1(0,0)$ & $\boldsymbol{X}_1(W_1^{(1)}, W_{2 \oplus 3}^{(1)})$ & $\boldsymbol{X}_1(W_1^{(2)}, W_{2 \oplus 3}^{(2)})$ & $\dotsm$ & $\boldsymbol{X}_1(W_1^{(B-1)}, W_{2 \oplus 3}^{(B-1)})$&  $\boldsymbol{X}_1(W_1^{(B)}, W_{2 \oplus 3}^{(B)})$ \\
Node 2 decodes &  $\hat{W}_1^{(1)},\hat{W}_3^{(1)}$ & $\leftarrow \; \hat{W}_1^{(2)},\hat{W}_3^{(2)}$ & $\leftarrow \; \hat{W}_1^{(3)},\hat{W}_3^{(3)}$ & $\dotsm$  & $\leftarrow \; \hat{W}_1^{(B)},\hat{W}_3^{(B)}$ & --- \\
Node 3 decodes &  $\hat{W}_1^{(1)},\hat{W}_2^{(1)}$ & $\leftarrow \; \hat{W}_1^{(2)},\hat{W}_2^{(2)}$ & $\leftarrow \; \hat{W}_1^{(3)},\hat{W}_2^{(3)}$ & $\dotsm$  & $\leftarrow \; \hat{W}_1^{(B)},\hat{W}_2^{(B)}$ & ---
\end{tabular}\\
\vspace{2ex}
\begin{flushleft}
Note: For simplicity the index denoting the node that decodes the messages are omitted. The arrows show the decoding order.
\end{flushleft}
\hrulefill
\end{table*}

From Lemma~\ref{lemma:awgn-lower}, we know that the equal-rate $R$ is achievable if $R < \log_2 ( 1 + \gamma_k)$ and $R < \frac{1}{2}\log_2 ( 1 + 2 \gamma_k)$, for all $k \in \{1,2,3\}$. Note that the second inequality implies the first inequality. Hence, if $R < \frac{1}{2} \log_2 \left( 1 + 2 \min_{k \in \{1,2,3\}}\gamma_k \right)$, then the equal rate $R$ is achievable. Taking the supremum of all $R$ yields the capacity~\eqref{eq:sender-symmetrical-capacity}.
\end{proof}

Note that the equal-rate capacity of the sender-symmetrical phase-fading AWGN three-way channel is achieved by the non-cooperative coding scheme, where feedback is not utilized.

\subsection{The Reciprocal Case}

Consider another class of phase-fading AWGN channels where the SNR on the link $i \rightarrow j$ equals that on the link $j \rightarrow i$, i.e., $\gamma_{ij} = \gamma_{ji}$ for all $i,j \in \{1,2,3\}$ where $i \neq j$. In this case, each receiver sees possibly different SNRs from different nodes (i.e., $\gamma_{ik}$ might not equal $\gamma_{jk}$), but for a pair of nodes, the SNR of the forward link is the same as that of the backward link. We term these channels the {\em reciprocal} phase-fading AWGN three-way channels.

Unlike the sender-symmetrical case, the non-cooperative coding scheme derived in Sec.~\ref{section:awgn} does not always achieve the equal-rate capacity upper bound for the reciprocal case. For example, let $\gamma_{23} = \gamma_{32} = 1$, $\gamma_{12} = \gamma_{21} = 6$, and $\gamma_{13} = \gamma_{31} = 8$. From Lemma~\ref{lemma:awgn-outer}, the equal-rate capacity is upper bounded as $C \leq \frac{1}{2} \log_2 8 = 1.5$. From Lemma~\ref{lemma:awgn-lower}, the non-cooperative scheme achieves rates up to $R < \log_2 2 = 1$.

We now propose a coding scheme that utilizes feedback via routing and network coding, and show that it achieves the equal-rate capacity upper bound (and hence the equal-rate capacity).
Without loss of generality, let $\gamma_{23} \leq \gamma_{12} \leq \gamma_{13}$. We have the following:

\begin{theorem} \label{theorem:equal-rate-reciprocal}
The equal-rate capacity of the reciprocal phase-fading AWGN three-way channel where $\gamma_{23} \leq \gamma_{12} \leq \gamma_{13}$ is 
\begin{equation}
C_\text{r} =  \frac{1}{2} \log \left( 1 + \gamma_{12} + \gamma_{23} \right).
\end{equation}
\end{theorem}

\begin{proof}[Proof of Thereom~\ref{theorem:equal-rate-reciprocal}]
It follows directly from Lemma~\ref{lemma:awgn-outer} that the equal-rate capacity for the reciprocal phase-fading AWGN three-way channel is upper bounded as $C_\text{r} \leq \frac{1}{2} \log \left( 1 + \gamma_{12} + \gamma_{23} \right)$.

We now propose a coding scheme that achieves this equal rate. Consider $B$ uniformly-distributed messages of equal size for each node, $W_i^{(b)} \in \{0,1,\dotsc, 2^{nR}-1\}$, for all $b \in \{1,2,\dotsc, B\}$ and for all $i \in \{1,2,3\}$. Each node sends its $B$ messages in $(B+1)$ blocks of $n$ channel uses. If the probability that any node wrongly decodes any message can be made arbitrarily small, the achievable equal rate $BR/(B+1) \rightarrow R$ as $B \rightarrow \infty$. Let $W_i(0)=W_i{(B+1)} =0$ be dummy messages, for all $i$. 

The encoding scheme is as follows:
\begin{enumerate}
\item In block $b \in \{1,2,\dotsc, B+1\}$: Node 2 transmits the codeword $\boldsymbol{X}_2(W_2^{(b)})$, where each codeletter is independently generated according to the circularly symmetric complex Gaussian distribution with variance $E[|X_2[t]|^2] = P$.
\item In block $b \in \{1,2,\dotsc, B+1\}$: Node 3 transmits the codeword $\boldsymbol{X}_3(W_3^{(b)})$,  where each codeletter is independently generated according to the circularly symmetric complex Gaussian distribution with variance $E[|X_3[t]|^2] = P$.
\item In block $b \in \{1,2,\dotsc, B+1\}$: Assuming that node 1 has decoded $W_2^{(b-1)}$  and $W_3^{(b-1)}$, it transmits the codeword $\boldsymbol{X}_1(W_1^{(b-1)}, W_{2\oplus 3}^{(b-1)})$,  where each codeletter is independently generated according to the circularly symmetric complex Gaussian distribution with variance $E[|X_1[t]|^2] = P$, and  $W_{2 \oplus 3}^{(b-1)} \triangleq W_2^{(b-1)} + W_3^{(b-1)} \mod 2^{nR}$ is the modulo-sum of the other nodes' messages in the previous block.
\end{enumerate}
The codewords for each blocks are generated independently. Note that there are all together $2^{nR}$ codewords for each of nodes 2 and 3 for each block. However, as node 1 uses a doubly-indexed codewords, there are $2^{2nR}$ codewords for each block. Also note that node 1 delays in transmitting its own message by one block.

The encoding and decoding process is depicted in Table~\ref{table}.

Node 1 decodes $W_2^{(b)}$ and $W_3^{(b)}$ at the end of block $b$. Using the results of the multiple-access channel, node 1 can reliably decodes its intended messages if
\begin{align}
R &< \log ( 1 + \gamma_{12}) \label{eq:reciprocal-node-1a} \\
R &< \log ( 1 + \gamma_{13})  \label{eq:reciprocal-node-1b}\\
2R &< \log ( 1 + \gamma_{12} + \gamma_{13} ). \label{eq:reciprocal-node-1c}
\end{align}
This sequence of decoding is necessary for node 1 as it transmits $W_{2 \oplus 3}^{(b)}$ in block $(b+1)$.

Node 2 employs simultaneous, {\em backward} decoding, meaning that it decodes the message in the reverse order after the entire $(B+1)$ blocks of transmissions~\cite[Ch.~7]{willems82}. More specifically, it decodes $W_1^{(b)}$ and $W_3^{(b)}$ starting from $b=B$, followed by $b=B-1$, and so on until $b=1$. It does so using (i) its received signals in blocks $b$ and $(b+1)$, (ii) its previously decoded messages, and (iii) its own messages, 

Assuming that it has correctly decoded $\{w_1^{(a)}, w_3^{(a)}\}_{a=b+1}^B$, and knowing its own messages $\{w_2^{(a)}\}_{a=1}^B$, it declares that $w_1^{(b)}=p$ and $w_3^{(b)}=q$ were sent if it finds a unique pair $(p,q)$ such that
\begin{multline}
\left( \boldsymbol{x}_1(p, \underline{w_2^{(b)}}+ q \text{ mod } 2^{nR}), \boldsymbol{x}_3(\underline{w_3^{(b+1)}}), \boldsymbol{y}_2'^{(b+1)} \right)\\ \in \mathcal{A}_\eta(X_1,X_3,Y'_2),  \label{eq:typical-1}
\end{multline}
and
\begin{equation}
\left( \boldsymbol{x}_3(q), \boldsymbol{y}_2'^{(b)} \right)\in \mathcal{A}_\eta(X_3,Y'_2), \label{eq:typical-2}
\end{equation}
where $\boldsymbol{y}_2'^{(b)}$ contains the received symbols $y_2[t]$, as well as the phase shifts values $(\theta_{12}[t],\theta_{32}[t])$, in block $b$, and $\mathcal{A}_\eta(\cdot)$ is the set of jointly typical sequences~\cite[p.~521]{coverthomas06}. Here we have underlined the messages that node 2 has already decoded or knows a priori when decoding $(w_1^{(b)},w_3^{(b)})$.

Assume that $w_1^{(b)}=p$ and $w_3^{(b)}=q$ are the actual messages sent, and let $p' \in \{1,2,\dotsc, 2^{nR}\} \setminus \{p\}$ and $q' \in \{1,2,\dotsc, 2^{nR}\} \setminus \{q\}$ be some wrong messages.
Node 2 makes a decoding error if any of the following event happens:\\
\indent \indent  {[}E1{]} $(p,q)$ does not satisfy \eqref{eq:typical-1} and \eqref{eq:typical-2}.\\
\indent \indent  {[}E2{]} some $(p',q)$ satisfies \eqref{eq:typical-1} and \eqref{eq:typical-2}.\\
\indent \indent   {[}E3{]} some $(p,q')$ satisfies \eqref{eq:typical-1} and \eqref{eq:typical-2}.\\
\indent \indent   {[}E4{]} some $(p',q')$ satisfies \eqref{eq:typical-1} and \eqref{eq:typical-2}. 

It follows from the joint asymptotic equipartition property~\cite[Thm.\ 15.2.1]{coverthomas06} that $\Pr\{\text{E1}\} < \eta$.

Now,
\begin{subequations}
\begin{align}
\Pr\{\text{E2}\} &\leq \Pr\{ \text{some } (p',q) \text{ satisfies } \eqref{eq:typical-1} \} \label{eq:e2-1} \\
& \leq \sum_{p' \in \{1,2,\dotsc, 2^{nR}\} \setminus \{p\}} 2^{-n[I(X_1;X_3,Y'_2) - 6\eta]} \label{eq:e2-2} \\
&= (2^{nR}-1) - 2^{-n[I(X_1;Y_2|X_3,\theta_{12}, \theta_{32}) - 6\eta]} \label{eq:e2-3} \\
&< 2^{n[ R - I(X_1;Y_2|X_3,\theta_{12}, \theta_{32}) + 6\eta]},\label{eq:e2-4}
\end{align}
\end{subequations}
where \eqref{eq:e2-1} follows from the union bound of probability, \eqref{eq:e2-2} follows from a property of joint typicality~\cite[Thm.~15.2.3]{coverthomas06}, \eqref{eq:e2-3} is derived because $X_1$ and $(X_3,\theta_{12}, \theta_{32})$ are independent.

Also, 
\begin{subequations}
\begin{align}
\Pr\{\text{E3}\} &= \sum_{q' \in \{1,2,\dotsc, 2^{nR}\} \setminus \{q\}} \Big[\Pr\{ (p,q') \text{ satisfies } \eqref{eq:typical-1}\} \nonumber \\
&\quad\quad\quad\quad\quad\quad\quad\quad \times \Pr\{q' \text{ satisfies } \eqref{eq:typical-2}\} \Big] \label{eq:e3-1}\\
&\leq (2^{nR}-1) 2^{-n[I(X_1;X_3,Y'_2)-6\eta]} 2^{-n[I(X_3;Y'_2)-6\eta]} \\
&< 2^{n [ R - I(X_1;Y_2|X_3,\theta_{12}, \theta_{32}) - I(X_3;Y_2|\theta_{12}, \theta_{32}) - 12 \eta]} \label{eq:e3-2} \\
& = 2^{n[R - I(X_1,X_3;Y_2|\theta_{12}, \theta_{32}) - 12\eta]},\label{eq:e3}
\end{align}
\end{subequations}
where \eqref{eq:e3-1} is derived because the codewords are generated independently in different blocks, \eqref{eq:e3-2} is derived because $X_1$ and $(X_3,\theta_{12}, \theta_{32})$ are independent, and so are $X_3$ and $(\theta_{12}, \theta_{32})$.

Using the same arguments, we have that
\begin{subequations}
\begin{align}
\Pr\{\text{E4}\} &= \sum_{p'} \sum_{q'} \Big[ \Pr \{ (p',q') \text{ satisfies }  \eqref{eq:typical-1} \} \nonumber \\
&\quad\quad\quad\quad\quad \times \Pr\{ q' \text{ satisfies } \eqref{eq:typical-2} \} \Big] \\
&= (2^{nR}-1) (2^{nR}-1) 2^{-n[I(X_1;X_3,Y'_2)-6\eta]} \nonumber \\
&\quad \times 2^{-n[I(X_3;Y'_2)-6\eta]} \\
&< 2^{n[2R - I(X_1,X_3;Y_2|\theta_{12}, \theta_{32}) - 12\eta]}.
\end{align}
\end{subequations}

So by choosing a sufficiently large $n$ and a sufficiently small $\eta$, if
\begin{align}
R < &\min \{ I(X_1;Y_2|X_3,\theta_{12}, \theta_{32}), \frac{1}{2}I(X_1,X_3;Y_2|\theta_{12}, \theta_{32})\} \nonumber \\
& = \min \{ \log ( 1 + \gamma_{12}) , \frac{1}{2} \log( 1 + \gamma_{12} + \gamma_{23} ) \}, \label{eq:reciprocal-node-2}
\end{align}
then $\Pr\{\text{E1, E2, E3, or E4}\} \leq \bigcup_{i=1}^4 \Pr\{\text{E}i\}$ can be made arbitrarily small, meaning that node 2 can reliably decode $(p,q)$.

Repeating the same argument for node 3, we can show that node 3 can reliably decode $W_1^{(b)}$ and $W_2^{(b)}$ if
\begin{align}
R &< \min \{ I(X_1;Y_3|X_2,\theta_{13}, \theta_{23}), \frac{1}{2}I(X_1,X_2;Y_3|\theta_{13}, \theta_{23})\} \nonumber \\
&= \min \{ \log ( 1 + \gamma_{13}) , \frac{1}{2} \log( 1 + \gamma_{13} + \gamma_{23} ) \}. \label{eq:reciprocal-node-3}
\end{align}

Now, note that $\log(1 + \gamma_{13}) \geq \log ( 1 + \gamma_{12}) \geq \log ( 1 + \frac{\gamma_{12}}{2} + \frac{\gamma_{23}}{2}) \geq \frac{1}{2} \log ( 1 + \gamma_{12} + \gamma_{23})$. So, if $R < \frac{1}{2} \log ( 1 + \gamma_{12} + \gamma_{23})$, then \eqref{eq:reciprocal-node-1a}, \eqref{eq:reciprocal-node-1b}, \eqref{eq:reciprocal-node-1c}, \eqref{eq:reciprocal-node-2}, and \eqref{eq:reciprocal-node-3} are satisfied.

Repeating the above analysis for all $b \in \{1,2,\dotsc, B\}$, and choosing a sufficiently large $B$, a much larger $n$, and a sufficiently small $\eta$, the equal rate arbitrarily close to $R$ is achievable.

Taking the supremum of the achievable equal rates, and comparing it with the upper bound, we have Theorem~\ref{theorem:equal-rate-reciprocal}.
\end{proof}

\section{Concluding Remarks}

We have derived the capacity region of the finite-field three-way channel, and the equal-rate capacity of the sender-symmetric and the reciprocal AWGN three-way channels.  The capacity results for the first two channel models are achieved using the non-cooperative (without using feedback) coding scheme, while that for the reciprocal channel is achieved using our proposed cooperative (using feedback) scheme.

In the cooperative coding scheme, node 1 transmits using doubly-indexed codewords. We now show that the equal-rate capacity is also achievable using superposition coding~\cite[p.~119]{elgamalkim2001} (a coding scheme commonly used for sending two independent messages), i.e., $\boldsymbol{X}_1 = \boldsymbol{U}(W_1) + \boldsymbol{V}(W_{2 \oplus 3})$ where the codeletters $U$ and $V$ are independently generated according to the circular symmetric complex Gaussian distribution with variances $\alpha P_1$ and $(1-\alpha)P_1$ respectively, for some $0 \leq \alpha \leq 1$. As an example, consider the decoding of the messages $W_1^{(b)}$ and $W_3^{(b)}$ by node 2, using its received signals in blocks $b$ and $(b+1)$:
\begin{align*}
\boldsymbol{Y}_2^{(b)} &= G_{12} [\boldsymbol{U}(W_1^{(b-1)}) + \boldsymbol{V}(W_{2 \oplus 3}^{(b-1)}) ]  + G_{32}\boldsymbol{X}_3(W_3^{(b)}) \nonumber \\&\quad + \boldsymbol{Z}_2^{(b)} \\
\boldsymbol{Y}_2^{(b+1)} &= G_{12} [\boldsymbol{U}(W_1^{(b)}) + \boldsymbol{V}(W_{2 \oplus 3}^{(b)}) ] + G_{32}\boldsymbol{X}_3(W_3^{(b+1)})  \nonumber \\&\quad + \boldsymbol{Z}_2^{(b+1)}.
\end{align*}

Using backward decoding, it would have decoded $W_3^{(b+1)}$, but not $W_1^{(b-1)}$ and $W_3^{(b-1)}$, which appear to be noise. Following the typical-set decoding arguments as in the proof of Theorem~\ref{theorem:equal-rate-reciprocal}, node 2 can reliably decode $(W_1^{(b)}, W_3^{(b)})$ if
\begin{align}
R &< \log( 1 + \alpha \gamma_{12}) \triangleq R'(\alpha) \label{eq:super-1} \\
R &< \log( 1 + (1-\alpha) \gamma_{12}) + \log ( 1 + \frac{\gamma_{23}}{1 + \gamma_{12}}) \triangleq R''(\alpha) \label{eq:super-2}\\
2R &<\log(1 + \gamma_{12} + \gamma_{23}). \label{eq:super-3}
\end{align}
where \eqref{eq:super-1} and \eqref{eq:super-2} ensure reliable decoding of $W_1^{(b)}$ and $W_3^{(b)}$ respectively assuming that the other messages has been decoded correctly, and \eqref{eq:super-3} ensures that both $W_1^{(b)}$ and $W_3^{(b)}$ are decoded reliably.
Comparing these rate bounds to those using double-index encoding [i.e., \eqref{eq:reciprocal-node-2}], we note that \eqref{eq:super-1} and \eqref{eq:super-2} appear to be more restrictive. However, maximizing the above rates over $\alpha$, we have that $\max_{0 \leq \alpha \leq 1} \min \{R'(\alpha), R''(\alpha) \} = \log ( 1 + \frac{(1+\gamma_{12})(\gamma_{12} + \gamma_{23})}{2 (1 + \gamma_{12}) + \gamma_{23}}) \triangleq R''' $. Numerical results show that $R''' \geq \frac{1}{2} \log(1 + \gamma_{12} + \gamma_{23})$ for all $0 \leq \gamma_{23} \leq \gamma_{12}$, suggesting that superposition coding is also able to achieve the equal-rate capacity. The sum-rate constraint is the ``active'' constraint for both double-index encoding and superposition encoding, and since the sum-rate constraint is the same for the two schemes, they achieve the same equal rate.

%The technique of double-index encoding and backward decoding were used for the relay channel~\cite[p.~119]{elgamalkim2001}. While this technique achieves the same rate as superposition coding and sliding-window decoding for the relay channel~\cite{kramergastpar04}, the former can achieve a higher equal rate in the AWGN three-way channel.

In the proposed cooperative scheme, node 1 facilitates the message exchange between nodes 2 and 3 by sending the ``network-coded'' messages $W_{2 \oplus 3}$,  while broadcasting its own message at the same time. At first sight, this seems similar to the two-way relay channel in which lattice codes can be used to improve the achievable rates over Gaussian codes~\cite{wilsonnarayananpfisersprintson10,namchunglee09}. However, the gain provided by using lattice codes in the two-way relay channel relies on the fact that the relaying node decodes the function $W_{2 \oplus 3}$ directly instead of the individual messages. This is undesirable in the three-way channel where node 1 must decode $W_2$ and $W_3$.

%\bibliography{../bib}


\end{document}